\documentclass[11pt]{article}

\usepackage[T1]{fontenc}
\usepackage[utf8]{inputenc}
\usepackage[british]{babel}
\usepackage[
rm={oldstyle=false,proportional=true},
sf={oldstyle=false,proportional=true},
tt={oldstyle=false,proportional=true}
]{cfr-lm}
\usepackage[babel]{microtype}

\usepackage{authblk}

\usepackage{mathrsfs}
\usepackage{mathtools}
\usepackage{amsfonts}
\usepackage{amssymb}
\usepackage{amsthm}

\usepackage{enumerate}

\usepackage{tikz}
\usetikzlibrary{decorations.pathreplacing}

\usepackage{xcolor}
\usepackage{array}
\usepackage{colortbl}
\usepackage{booktabs}
\usepackage{multirow}

\newtheorem{theorem}{Theorem}[section]
\newtheorem{lemma}[theorem]{Lemma}

\newtheorem{corollary}[theorem]{Corollary}

\theoremstyle{definition}
\newtheorem{definition}[theorem]{Definition}
\newtheorem{problem}[theorem]{Open problem}

\usepackage{tikz}
\usetikzlibrary{decorations.pathmorphing}
\usetikzlibrary{graphs}
\usetikzlibrary{graphs.standard}
\usetikzlibrary{quotes}

\usepackage[misc]{ifsym}

\usepackage[noadjust]{cite}
\usepackage[hidelinks,bookmarksnumbered]{hyperref}
\usepackage[capitalise,noabbrev]{cleveref}

\usepackage{orcidlink}

\newcommand{\cclass}[1]{\ensuremath{\mathord{\textrm{#1}}}} %from Ben Hescott! :)

\newcommand{\ruleset}[1]{\textsc{#1}}
\newcommand{\rsMisParck}{\ruleset{Mis\`ere PArcK}}
\newcommand{\rsMisPartizanArcK}{\ruleset{Mis\`ere Partizan Arc Kayles}}
\newcommand{\rsPosCNF}{\ruleset{PosCNF}}
\newcommand{\rsPositiveCNF}{\ruleset{PositiveCNF}}
\newcommand{\rsBoundedCL}{\ruleset{Bounded Two-Player Constraint Logic}}
\newcommand{\rsBCL}{\ruleset{B2CL}}

\newcounter{authcount}

\NewDocumentCommand{\authDetails}{m m m o}{%
    \stepcounter{authcount}%
    \IfNoValueTF{#4}{%
        \author[\arabic{authcount}]%
    }%
    {%
        \author[#4]%
    }%
    {%
        \mbox{#1\,$^{\textrm{\href{mailto:#2}{\Letter}}\,\,%
        \ifx&#3&\else\raisebox{-0.2ex}{\orcidlink{#3}}\,\fi}$}%
    }%
}

\title{\rsMisPartizanArcK{} is \cclass{PSPACE}-complete, even on Planar Graphs}

\authDetails{Kyle Burke}{kburke@flsouthern.edu}{}[1]
\authDetails{Caroline Cashman}{}{}[2]
\authDetails{Alfie Davies}{research@alfied.xyz}{0000-0002-4215-7343}[3]
\authDetails{Kanae Yoshiwatari}{yoshiwatari.kanae@amp.i.kyoto-u.ac.jp}{0000-0001-5259-7644}[4]
\authDetails{Francesca Yu}{}{}[5]

\affil[1]{ Florida Southern College\\Lakeland\\USA}
\affil[2]{ College of William \& Mary\\Williamsburg\\USA}
\affil[3]{ Memorial University of Newfoundland\\Canada}
\affil[4]{ Kyoto University\\Japan}
\affil[5]{ University of California, Berkeley\\USA}

\date{}

\begin{document}

\maketitle

\begin{abstract}
    \noindent
    We show that \rsMisPartizanArcK{} is \cclass{PSPACE}-complete on planar
    graphs via a reduction from \rsBoundedCL{}.  Furthermore, we show how to embed our
    gadgets onto the square and triangular grids.  In order to clearly explain these results, we get into
    the details of \rsBoundedCL{} and find three \cclass{PSPACE}-complete variants of that as well.
\end{abstract}

\section{Introduction}

Combinatorial Games are two-player games with no randomness and no hidden
information where the players (left and right) alternate turns according to the
rules until the game ends.  A \emph{position} is a game state, and the
positions that the left and right players can move to are the left and right
\emph{options}.  For many rulesets, the default \emph{normal play} convention
is used: when the set of options a player has from the current position is
empty, that player loses the game.  Under \emph{mis\`ere play}, when a player
has no options available, they win instead.  A combinatorial game ruleset is
\emph{impartial} if the two players have the same options from every position.
Otherwise, the game is strictly \emph{partizan}.  \cite{WinningWays:2001}, \cite{LessonsInPlay:2007}, \cite{SiegelCGT:2013}.

\ruleset{Arc Kayles} is an impartial combinatorial game played on a simple
undirected graph.  Each turn a player chooses one of the edges on the graph,
$(u, v)$, then removes both vertices $u$ and $v$ as well as all edges incident
to either.  It is played with the normal play convention.

\ruleset{Partizan Arc Kayles} is a partizan game just like \ruleset{Arc
Kayles}, except that all edges in the graph are coloured Blue or Red.  On their
turn, the left player can only pick a blue edge and the right player can only
pick a red edge, though edges of both colours adjacent to those are deleted.

\ruleset{Mis\`ere Partizan Arc Kayles} is the same as \ruleset{Partizan Arc
Kayles}, except that it is played with the mis\`ere play convention.  Thus the
winner of a game is the first player to start their turn with no edges of their
colour.  We provide a playable version of \rsMisParck{} at
\url{https://kyleburke.info/DB/combGames/miserePartizanArcKayles.html}.

Computational Complexity provides one means to assess how fun a game is.  For games
like \ruleset{Undirected Geography} where the winning strategy can be
determined in polynomial time \cite{DBLP:journals/tcs/FraenkelSU93}, there is
no meaningful competition if one or both players know how to choose a winning option.
If there are polynomial options from any position, then by using the process to
determine winnability a player can select a winning move if one exists in polynomial time.

Instead, evidence that there may not be an efficient algorithm to solve a game
adds to the intensity of the competition as players use uncertain heuristics to
make their moves.  \ruleset{Generalized Chess}, for example, is
\cclass{EXPTIME}-complete, so no polynomial-time algorithm exists to determine
the winnability from any position.  Many loop-free games have been shown to be
\cclass{PSPACE}-complete, meaning that unless $\cclass{P} = \cclass{PSPACE}$,
no polynomial time algorithm exists for them either
\cite{PapadimitriouBook:1994}, e.g. \ruleset{Hex} \cite{Reisch:1981},
\ruleset{Amazons} \cite{DBLP:books/daglib/0023750}, and \ruleset{Geography}
\cite{DBLP:journals/jcss/Schaefer78}.  The intractability of these games
increases their replayability as players hone their imperfect strategies.
Unlike many other scientific fields, showing that a ruleset is computationally hard
is beneficial and can help explain (or even promote) its popularity.

The computational complexity of \ruleset{Partizan Arc Kayles} is of significant
interest because it is a general version of the game \ruleset{Domineering}, for
which the complexity remains unsolved despite multiple
attempts\footnote{Finding \ruleset{Domineering} to be \cclass{PSPACE}-complete
would prove the same for \ruleset{Partizan Arc Kayles}, but not the other way
around.}.

In sections \ref{sec:overview} and \ref{sec:reduction}, we prove that
\rsMisParck{} is \cclass{PSPACE}-complete.  First, in section
\ref{sec:overview}, we provide an overview of the reduction, including a
description of the ruleset  \rsPosCNF{} and a thorough discussion of \rsBCL{}
and some variants, which we reduce directly from.  (This includes an
explanation of why we cannot reduce from \rsPosCNF{} directly.)  In Section
\ref{sec:reduction}, we describe the individual gadgets for the reduction to
complete the proof of hardness.  We provide gadgets for both square and
triangular grids, showing that the game is \cclass{PSPACE}-complete in both more-specific
cases.

\section{Reduction Overview}
\label{sec:overview}

The reduction we present to show \rsMisParck{} is \cclass{PSPACE}-complete is
from \rsPositiveCNF{} (\rsPosCNF{}) via \rsBoundedCL{} (\rsBCL).  Much of the
terminology stems from the \rsPosCNF{} part, but the detour through \rsBCL{}
provides us with the benefit of planarity.  This is a common path for proving
game complexity.  We include this section to help explain this for readers not
used to these constructions.

\subsection{\texorpdfstring{\rsPositiveCNF{}}{PositiveCNF}}

\rsPositiveCNF{} (\rsPosCNF) is a boolean-formula game with players named True
and False.  The position consists of a conjunctive normal form (CNF) boolean
formula without negations (``positive'') and the current state of the
variables, with each being either \emph{true}, \emph{false}, or
\emph{unassigned}.  On their turn, a player chooses one of the unassigned
variables and assigns their value to it.  After all variables have been
assigned, the winner is the player whose identity is what the formula evaluates
to.  Since the formula is positive, it would never benefit a player to assign
their opponent's value to a variable.

Aside from the positive nature of the formula, there are some important
distinctions between the Quantified Boolean Formula (game) and \rsPosCNF{}:
\begin{itemize}
    \item \rsPosCNF{} is not known to be hard on only 3-CNF formulas, so some
        of the clauses may be much longer in order for it to be computationally hard.
    \item The order of variables is not decided ahead of time.  The difficulty
        of each turn is not in picking the value to assign, but in choosing which of
        the unassigned variables to claim.
\end{itemize}

\subsection{\texorpdfstring{\rsBoundedCL{}}{Bounded Two-Player Constraint Logic}}

\rsBoundedCL{} (\rsBCL) is a game where each turn consists of flipping an edge
in a directed graph.  The position consists of a digraph of coloured (Blue,
Red) arcs that each have weight 1 or 2.  Each edge can only be flipped once, so
in a given position, some edges might have already been flipped.  One blue arc
and one red arc are each indicated as the goals for that player.  Each vertex
must have a total in-weight of at least 2.  A turn consists of picking an arc
of your colour and flipping it without violating the in-weight condition.  If a
player flips their goal arc, then the game ends and they are declared the
winner. If a player is ever unable to move, then the game ends in a draw.

\rsBCL{} is defined and shown to be \cclass{PSPACE}-hard in
\cite{DBLP:books/daglib/0023750}.  Note: we use a different description of the
game. In the original work, the edges are coloured White and Black to identify
the two players (Blue and Red for us) and colours are instead used to
identify the weight of the edges.  (Red for 1 and blue for 2.)  In our diagrams
of \rsBCL{}, we use extra arrowheads to indicate the edge weights, and just color the edges by their player.

The reduction is from \rsPosCNF{} and includes clever tricks so that the only
required vertices to show hardness are encodings of And, Or, Fanout, Choice,
and Variable gadgets; all other vertices can be transformed into combinations
of those five.  The reduced graphs are designed so that one of the goal edges
will be triggered before either player can't move, avoiding the issue of draws.
Additionally, a crossover gadget can be built from those five so that the
planar version of \rsBCL{} is also \cclass{PSPACE}-hard.

\begin{figure}[ht]
    \begin{center}
    \begin{tikzpicture}[node distance = 1.2cm, minimum size = .2cm, inner sep = .07cm, ultra thick]
        \node[circle, draw, fill=black] (andC) at (0, 0) {};
        \node[] (andL) [below left of=andC] {};
        \node[] (andT) [above of=andC] {};
        \node[] (andR) [below right of=andC] {};
        \node[] (andB) [below of=andC] {And};

        \path[->]
            (andC) edge [color = blue] (andR)
            (andC) edge [color = blue] (andL)
        ;
        \path[->>]
            (andT) edge [color = blue] (andC)
        ;

        \node[circle, draw, fill=black] (orC) at (2, 0) {};
        \node[] (orL) [below left of=orC] {};
        \node[] (orT) [above of=orC] {};
        \node[] (orR) [below right of=orC] {};
        \node[] (orB) [below of=orC] {Or};

        \path[->]
        ;
        \path[->>]
            (orC) edge [color = blue] (orR)
            (orC) edge [color = blue] (orL)
            (orT) edge [color = blue] (orC)
        ;

        \node[circle, draw, fill=black] (choiceC) at (4, 0) {};
        \node[] (choiceL) [above left of=choiceC] {};
        \node[] (choiceT) [above of=choiceC] {};
        \node[] (choiceR) [above right of=choiceC] {};
        \node[] (choiceB) [below of=choiceC] {Choice};

        \path[->]
            (choiceC) edge [color = blue] (choiceB)
            (choiceL) edge [color = blue] (choiceC)
            (choiceR) edge [color = blue] (choiceC)
        ;
        \path[->>]
        ;

        \node[circle, draw, fill=black] (fanoutC) at (6, 0) {};
        \node[] (fanoutL) [above left of=fanoutC] {};
        \node[] (fanoutT) [above of=fanoutC] {};
        \node[] (fanoutR) [above right of=fanoutC] {};
        \node[] (fanoutB) [below of=fanoutC] {Fanout};

        \path[->]
            (fanoutL) edge [color = blue] (fanoutC)
            (fanoutR) edge [color = blue] (fanoutC)
        ;
        \path[->>]
            (fanoutC) edge [color = blue] (fanoutB)
        ;

        \node[circle, draw, fill=black] (variableC) at (3, -2.5) {};
        \node[] (variableL) [left of=variableC] {};
        \node[] (variableR) [right of=variableC] {};
        \node[] (variableB) at (3, -3.2) {Variable};

        \path[->]
        ;
        \path[->>]
            (variableL) edge [color = blue] (variableC)
            (variableR) edge [color = red] (variableC)
        ;
    \end{tikzpicture}
    \caption{B2CL Gadgets.}
    \label{fig:b2clGadgets}
    \end{center}
\end{figure}
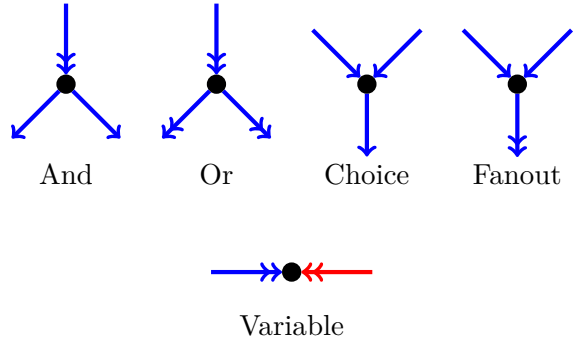

Since \rsBCL{} isn't a boolean formula game, the gadgets are described using
the terms Active and Inactive in place of True and False, respectively.  In
these \rsPosCNF{}-based games, the difficult decisions for the players are
still in the order they choose the variables.  After that, Blue spends their
moves appropriately propagating the active and inactive signals through the
all-blue (And, Or, Fanout, and Choice) vertices to see whether they can flip
their goal edge.

As with \rsPosCNF{}, it never hurts Blue to have an Active output from one
vertex gadget instead of an Inactive with all else being the same.

\cite{DBLP:books/daglib/0023750} hints that instead of using their goal edges,
``it is normal to define the loser as the first player unable to move.''  We
spell that out here, both for Normal Play and an intermediate game where the
players take on roles similar to those in Maker-Breaker games\footnote{This
will not be an actual Maker-Breaker game, as those have additional
requirements, e.g. coming from a positional game.}.

\newcommand{\rsBBBCLFull}{\ruleset{Builder-Blocker} \rsBoundedCL{}}
\newcommand{\rsBBBCL}{\ruleset{BBB2CL}}

\begin{definition}[\rsBBBCLFull{}]
    \textsl{\rsBBBCLFull{}} (\textsl{\rsBBBCL{}}) is the same game as
    \textsl{\rsBCL{}}, except that there is no Red goal edge.  If Blue flips
    the(ir) goal edge, the game ends and they win.  If the goal edge has not
    been flipped and the current player can't make a move, then instead the
    game ends with Red as the winner.
\end{definition}

\newcommand{\rsNPBCLFull}{\ruleset{Normal-Play} \rsBoundedCL{}}
\newcommand{\rsNPBCL}{\ruleset{NPB2CL}}

\begin{definition}[\rsNPBCLFull{}]
    \textsl{\rsNPBCLFull{}} (\textsl{\rsNPBCL{}}) is the same game as
    \textsl{\rsBBBCL{}}, except that there is no Blue goal edge.  When a player
    is unable to flip an edge, the game ends with the other player as the
    winner.
\end{definition}

We next show that these are both \cclass{PSPACE}-complete.  Our reductions will
piggyback off of the reduction from \cite{DBLP:books/daglib/0023750} with
slight modifications. An example of that reduction from \rsPosCNF{} is
included in Figure \ref{fig:b2clHardness}.

\begin{figure}[ht]
    \begin{center}
    \begin{tikzpicture}[node distance = 1.2cm, minimum size = .2cm, inner sep = .07cm, ultra thick]
        \node[circle, draw, fill=black] (w) at (0, 0) [label=left:{w}] {};
        \node[] (wDown) [below of=w] {};
        \node[circle, draw, fill=black] (wUp) [above of=w] {};

        \node[circle, draw, fill=black] (x) at (2, 0) [label=left:{x}] {};
        \node[] (xDown) [below of=x] {};
        \node[circle, draw, fill=black] (xUp) [above of=x] {};

        \node[circle, draw, fill=black] (y) at (4, 0) [label=left:{y}] {};
        \node[] (yDown) [below of=y] {};
        \node[circle, draw, fill=black] (yUp) [above of=y] {};

        \node[circle, draw, fill=black] (z) at (6, 0) [label=left:{z}] {};
        \node[] (zDown) [below of=z] {};
        \node[circle, draw, fill=black] (zUp) [above of=z] {};

        \node[circle, draw, fill=black] (g) at (8, 0) [] {};
        \node[] (gDown) [below of=g] {};
        \node[] (gUp) [above of=g] {};

        \node[circle, draw, fill=black] (wa) at (0.1, 2.2) {};
        \node[circle, draw, fill=black] (worx) at (0.2, 3.2) {};
        \node[circle, draw, fill=black] (xa) at (1.5, 2.2) {};
        \node[circle, draw, fill=black] (xb) at (4, 2.2) {};
        \node[circle, draw, fill=black] (wb) at (1.1, 3.2) {};
        \node[circle, draw, fill=black] (worxory) at (1.0, 4.0) {};
        \node[circle, draw, fill=black] (worxoryB) at (1.0, 4.8) {};
        \node[circle, draw, fill=black] (za) at (5.2, 2.2) {};
        \node[circle, draw, fill=black] (zb) at (6.5, 2.2) {};
        \node[circle, draw, fill=black] (worz) at (4, 4.0) {};
        \node[circle, draw, fill=black] (worzB) at (4, 4.8) {};
        \node[circle, draw, fill=black] (xorz) at (6, 4.0) {};
        \node[circle, draw, fill=black] (xorzB) at (6, 6.0) {};
        \node[circle, draw, fill=black] (worxoryAndworz) at (2.5, 5.5) {};
        \node[circle, draw, fill=black] (worxoryAndworzB) at (4.0, 6.0) {};
        \node[circle, draw, fill=black] (wholeForm) at (5.0, 6.5) {};
        \node[circle, draw, fill=black] (blueWin) at (5.0, 7.2) {};

        \node[circle, draw, fill=black] (xRedA) at (8, 3) [] {};
        \node[circle, draw, fill=black] (xRedB) [above of=xRedA, label=right:{\phantom{xx}$\times k$}] {};
        \node[circle, draw, fill=black] (xRedC) [above of=xRedB] {};

        \path[->]
            (wa) edge [color = blue] (wUp)
            (xa) edge [color = blue] (xUp)
            (xb) edge [color = blue] (xUp)
            (wb) edge [color = blue] (wUp)
            (za) edge [color = blue] (zUp)
            (zb) edge [color = blue] (zUp)
            (worxoryAndworz) edge [color = blue] (worxoryB)
            (worxoryAndworz) edge [color = blue] (worzB)
            (wholeForm) edge [color = blue] (worxoryAndworzB)
            (wholeForm) edge [color = blue] (xorzB)
        ;
        \path[->>]
            (wDown) edge [color = red] (w)
            (wUp) edge [color = blue] (w)
            (xDown) edge [color = red] (x)
            (xUp) edge [color = blue] (x)
            (yDown) edge [color = red] (y)
            (yUp) edge [color = blue] (y)
            (zDown) edge [color = red] (z)
            (zUp) edge [color = blue] (z)

            (worx) edge [color = blue] (wa)
            (worx) edge [color = blue] (xa)
            (worxory) edge [color = blue] (worx)
            (worxory) edge [color = blue] (yUp)
            (worz) edge [color = blue] (wb)
            (worz) edge [color = blue] (za)
            (xorz) edge [color = blue] (xb)
            (xorz) edge [color = blue] (zb)
            (worxoryB) edge [color = blue] (worxory)
            (worzB) edge [color = blue] (worz)
            (worxoryAndworzB) edge [color = blue] (worxoryAndworz)
            (xorzB) edge [color = blue] (xorz)
            (blueWin) edge [color = blue] node [left] {Blue goal\phantom{x}} (wholeForm)

            (gDown) edge [color = red] node [right] {Red goal} (g)
            (g) edge [color = blue] (gUp)

            (xRedA) edge [color = red] (xRedB)
            (xRedC) edge [color = red] (xRedB)
        ;
    \end{tikzpicture}
    \caption{
        Translated from Figure 6.1 of \cite{DBLP:books/daglib/0023750}, a
        \rsBCL{} position equivalent to the \rsPosCNF{} game on $(w \vee x \vee
        y) \wedge (w \vee z) \wedge (x \vee z)$.  For space reasons (as in the
        source figure) it does not contain all simplifications to only the five
        basis vertices and does not employ the crossover gadgets.  We do,
        however, include the extra $k$ hidden red plays described.
    }
    \label{fig:b2clHardness}
    \end{center}
\end{figure}
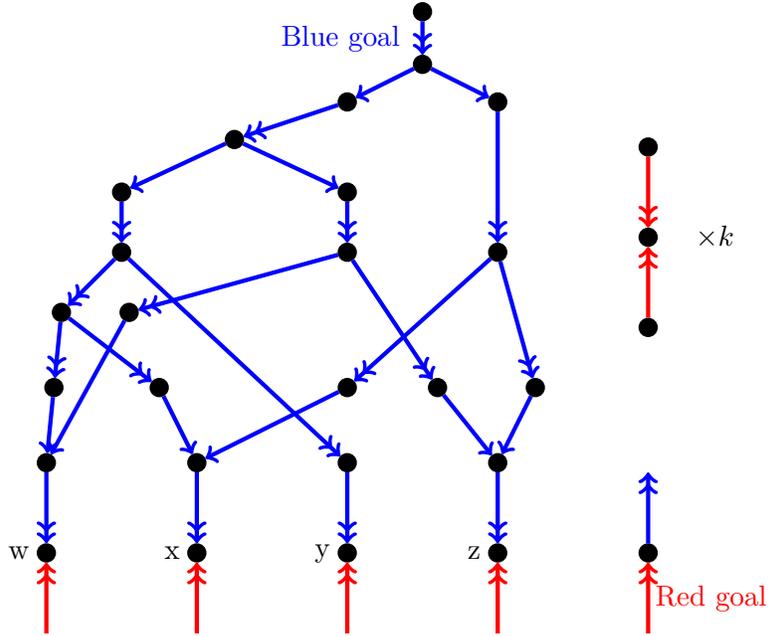

%%% I took these out because we haven't done them.
%TODO: Should we show how to build the $k$ extra red edges using the five basis vertices?

%Alfie: I think so, yes.

There are four main parts to the construction:
\begin{itemize}
    \item The bottom row of variable gadgets.
    \item The all-blue circuit-like gadgetry leading from the variables to (and
        including) the goal.
    \item The $k$ extra plays for Red.  $k$ should be chosen so that it is at
        least as large as the number of plays Blue has in the circuitry.
    \item The Red goal edge, adjacent to a final Blue edge that Blue will have
        to flip if they are unable to flip their goal edge, handing the win to
        Red.
\end{itemize}
Together, this assures that Blue wins exactly when they can choose variables to
satisfy the original positive CNF formula.

\rsBCL{} is used extensively to show \cclass{PSPACE}-hardness in combinatorial
games for three main reasons:
\begin{enumerate}
    \item There are only five basis vertices that are needed to include instead
        of every possible Constraint Logic vertex.
    \item Those five gadgets can be abstracted into their named meaning; often
        the actual rules for constraint logic (e.g. edge flipping) are (safely)
        overlooked.
    \item The crossover gadget can be built from the other five basis vertices,
        delivering planarity automatically.
\end{enumerate}

We continue by showing the necessary parts so that \rsBBBCL{} and \rsNPBCL{}
can also be used for the same purposes.

\begin{lemma}
    \label{lem:bbbcl}
    \textsl{\rsBBBCLFull{}} is \cclass{PSPACE}-complete.
\end{lemma}

\begin{proof}
    We use the same construction from \cite{DBLP:books/daglib/0023750} (example
    in Figure \ref{fig:b2clHardness}) except that the connected component with
    the Red goal is missing and we double the number of extra red plays ($2k$
    instead of $k$).  Since Red has so many moves, the game will not end due to
    them running out of plays.  Red will play on as many variables as possible,
    then just make their moves waiting to see whether Blue is able to activate
    the goal.  If Blue can, then they will win.  Otherwise, they will run out of
    moves and thus the game will end in Red's favour.
\end{proof}

\begin{corollary}
    \cclass{PSPACE}-hard instances of \textsl{\rsBBBCL{}} can be constructed
    from only the same five basis vertices as in \textsl{\rsBCL{}}.
\end{corollary}

\begin{proof}
    Our reduction is the same except that we remove one component and add
    $k$ more copies of a vertex that already existed in the original
    reduction.
\end{proof}

\begin{lemma}
    \label{lem:npbcl}
    \textsl{\rsNPBCLFull{}} is \cclass{PSPACE}-complete.
\end{lemma}

\begin{proof}
    We use the same construction as in lemma \ref{lem:bbbcl}, except that
    instead of a single goal edge for Blue, this is the start of a Blue chain
    of length $2k$.  Now, if Blue is able to flip what was the goal edge, they
    can now earn $2k$ extra moves, enough to continue playing after Red uses
    what remains of their $2k$ moves.
\end{proof}

\begin{corollary}
    \cclass{PSPACE}-hard instances of \textsl{\rsNPBCL{}} can be constructed
    from only the same five basis vertices as in \textsl{\rsBCL{}}.
\end{corollary}

\begin{proof}
    Our reduction is the same as for \textsl{\rsBBBCL{}}, except that we've
    added a long chain for Blue after what was their goal edge.  That chain can
    be created from connected Or gadgets, as we demonstrate in \cref{fig:npb2clHardness}.
\end{proof}

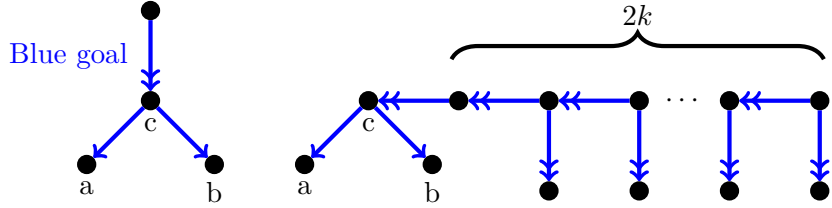
\begin{figure}[ht]
    \begin{center}
    \begin{tikzpicture}[node distance = 1.2cm, minimum size = .2cm, inner sep = .07cm, ultra thick]
        \node[circle, draw, fill=black] (a) at (0, 0) [label=below:{a}] {};
        \node[circle, draw, fill=black] (c) [above right of=a, label=below:{c}] {};
        \node[circle, draw, fill=black] (b) [below right of=c, label=below:{b}] {};
        \node[circle, draw, fill=black] (top) [above of=c] {};
        \node[circle, draw, fill=black] (ab) [right of=b, label=below:{a}] {};
        \node[circle, draw, fill=black] (cb) [above right of=ab, label=below:{c}] {};
        \node[circle, draw, fill=black] (bb) [below right of=cb, label=below:{b}] {};
        \node[circle, draw, fill=black] (k0) [right of=cb] {};
        \node[circle, draw, fill=black] (k1) [right of=k0] {};
        \node[circle, draw, fill=black] (k1b) [below of=k1] {};
        \node[circle, draw, fill=black] (k2) [right of=k1] {};
        \node[circle, draw, fill=black] (k2b) [below of=k2] {};
        \node[circle, draw, fill=black] (k3) [right of=k2] {};
        \node[circle, draw, fill=black] (k3b) [below of=k3] {};
        \node[circle, draw, fill=black] (k4) [right of=k3] {};
        \node[circle, draw, fill=black] (k4b) [below of=k4] {};

        \path[->]
            (c) edge [color = blue] (a)
            (c) edge [color = blue] (b)
            (cb) edge [color = blue] (ab)
            (cb) edge [color = blue] (bb)
        ;
        \path[->>]
            (top) edge [color = blue] node [left] {Blue goal\phantom{x}} (c)
            (k0) edge [color = blue] (cb)
            (k1) edge [color = blue] (k0)
            (k1) edge [color = blue] (k1b)
            (k2) edge [color = blue] (k1)
            (k2) edge [color = blue] (k2b)
            (k3) edge [color = blue] (k3b)
            (k4) edge [color = blue] (k3)
            (k4) edge [color = blue] (k4b)
        ;
        \path[]
            (k3) edge [draw=none] node {$\cdots$} (k2)
        ;
        \draw [decorate, decoration = {brace, amplitude=10pt}] (4.85, 1.4) --  (9.8, 1.4) node [pos=0.5,above=10pt,black]{$2k$};

    \end{tikzpicture}
    \caption{
        Modification of the Blue goal from the reduction to reach \rsNPBCL{}.  On the left side is the goal part of the original \rsBCL{} and \rsBBBCL{} reductions.  On the right side is the modification using additional Or gadgets that gives Blue $2k$ additional moves.
    }
    \label{fig:npb2clHardness}
    \end{center}
\end{figure}

This second result means that future results can build off of this without
showing a separate goal gadget.

Since the main topic of this paper is a mis\`ere game, we include a mis\`ere version of
\rsBCL{} and prove the analogous results about that.

\newcommand{\rsMPBCLFull}{\ruleset{Mis\`ere Play} \rsBoundedCL{}}
\newcommand{\rsMPBCL}{\ruleset{MPB2CL}}

\begin{definition}[\rsMPBCLFull{}]
    \textsl{\rsMPBCLFull{}} (\textsl{\rsMPBCL{}}) is the same game as
    \textsl{\rsNPBCL{}}, except that when a player is unable to flip an edge on
    their turn, they win the game.
\end{definition}

\begin{lemma}
    \textsl{\rsMPBCLFull{}} is \cclass{PSPACE}-complete.
\end{lemma}

\begin{proof}
    We use the same construction as in lemma \ref{lem:npbcl}, except that
    instead of the $2k$ long chain for Blue at the end of the circuitry, those
    chain edges are now Red.  We provide a separate chain of $2k$ moves for
    Blue.  The following prescribed order of the game is best for both players.

    Since this is mis\`ere, both players are looking to finish up their moves
    as quickly as possible.  The first player who can't move will win.

    Both players will first select all variables.  This is because Blue wants
    to open up the new $2k$ long chain for Red at the end of their circuitry
    section.  Red has the opposite goal, so they will also want to play on the
    variables first.  Afterwards, Blue will move through the circuitry to
    finish that before Red completes their original batch of $2k$ moves.  If
    Blue is able to activate the end of the circuit, then Red has an
    additional $2k$ moves.  Red had $k$ left from their original $2k$.  Red's total of $3k$ is greater than Blue's $2k$, so Blue will run out before Red and will win.  If Blue is not
    able to activate the end of the circuit, then Blue will still have $2k$
    moves plus whatever remains in the circuitry and Red will win by finishing
    their original $2k$ moves first.
\end{proof}

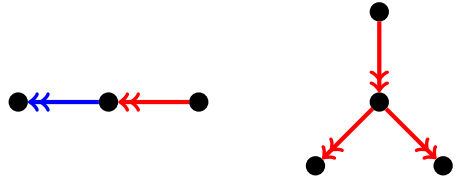
\begin{figure}[ht]
    \begin{center}
    \begin{tikzpicture}[node distance = 1.2cm, minimum size = .2cm, inner sep = .07cm, ultra thick]
        \node[circle, draw, fill=black] (a) {};
        \node[circle, draw, fill=black] (b) [right of=a] {};
        \node[circle, draw, fill=black] (c) [right of=b] {};
        \node[] (d) [right of=c] {};
        \node[circle, draw, fill=black] (e) [right of=d] {};
        \node[circle, draw, fill=black] (f) [below left of=e] {};
        \node[circle, draw, fill=black] (g) [below right of=e] {};
        \node[circle, draw, fill=black] (h) [above of=e] {};

        \path[->]
        ;
        \path[->>]
            (b) edge [color = blue] (a)
            (c) edge [color = red] (b)
            (e) edge [color = red] (f)
            (e) edge [color = red] (g)
            (h) edge [color = red] (e)
        ;
        \path[]
        ;

    \end{tikzpicture}
    \caption{
        New basis vertices needed for the reduction: Blue-to-Red (left) and Red-Or (right).
    }
    \label{fig:newBasisVertices}
    \end{center}
\end{figure}

\begin{corollary}
    \cclass{PSPACE}-hard instances of \textsl{\rsMPBCL{}} can be constructed
    from the five basis vertices as in \textsl{\rsBCL{}} plus two additional vertices: \emph{Blue-to-Red} and \emph{Red-Or}, as shown in \cref{fig:newBasisVertices}.
\end{corollary}

\begin{proof}
    %(Broken)

    Our reduction is the same as for \textsl{\rsNPBCL{}}, except for two changes that necessitate the new vertices: (1) the blue weight-2 edge that, if flipped, allows a red weight-2 edge to be flipped, and (2) a red version of the Or vertex.  These two form the new tail that trails from where the original Blue goal edge stood.  %that we've replaced a long chain of Blue edges with Red edges.  TODO: well, except the beginning has Red and Blue, so we need to show that that is still the case, actually!
\end{proof}

\begin{figure}[ht]
    \begin{center}
    \begin{tikzpicture}[node distance = 1.2cm, minimum size = .2cm, inner sep = .07cm, ultra thick]
        \node[circle, draw, fill=black] (a) at (0, 0) [label=below:{a}] {};
        \node[circle, draw, fill=black] (c) [above right of=a, label=below:{c}] {};
        \node[circle, draw, fill=black] (b) [below right of=c, label=below:{b}] {};
        \node[circle, draw, fill=black] (top) [above of=c] {};
        \node[circle, draw, fill=black] (ab) [right of=b, label=below:{a}] {};
        \node[circle, draw, fill=black] (cb) [above right of=ab, label=below:{c}] {};
        \node[circle, draw, fill=black] (bb) [below right of=cb, label=below:{b}] {};
        \node[circle, draw, fill=black] (kn1) [right of=cb] {};
        \node[circle, draw, fill=black] (k0) [right of=kn1] {};
        \node[circle, draw, fill=black] (k1) [right of=k0] {};
        \node[circle, draw, fill=black] (k1b) [below of=k1] {};
        \node[circle, draw, fill=black] (k2) [right of=k1] {};
        \node[circle, draw, fill=black] (k2b) [below of=k2] {};
        \node[circle, draw, fill=black] (k3) [right of=k2] {};
        \node[circle, draw, fill=black] (k3b) [below of=k3] {};
        \node[circle, draw, fill=black] (k4) [right of=k3] {};
        \node[circle, draw, fill=black] (k4b) [below of=k4] {};

        \path[->]
            (c) edge [color = blue] (a)
            (c) edge [color = blue] (b)
            (cb) edge [color = blue] (ab)
            (cb) edge [color = blue] (bb)
        ;
        \path[->>]
            (top) edge [color = blue] node [left] {Blue goal\phantom{x}} (c)
            (kn1) edge [color = blue] (cb)
            (k0) edge [color = blue] (kn1)
            (k1) edge [color = red] (k0)
            (k1) edge [color = red] (k1b)
            (k2) edge [color = red] (k1)
            (k2) edge [color = red] (k2b)
            (k3) edge [color = red] (k3b)
            (k4) edge [color = red] (k3)
            (k4) edge [color = red] (k4b)
        ;
        \path[]
            (k3) edge [draw=none] node {$\cdots$} (k2)
        ;
        \draw [decorate, decoration = {brace, amplitude=10pt}] (6.05, 1.4) --  (11.1, 1.4) node [pos=0.5,above=10pt,black]{$2k$};

    \end{tikzpicture}
    \caption{
        Modification of the Blue goal from the reduction to reach \rsMPBCL{}.  On the left side is the goal part of the original \rsBCL{} and \rsBBBCL{} reductions.  On the right side is the modification using additional gadgets that gives Red $2k$ moves if Blue would have met the goal originally.
    }
    \label{fig:mpb2clHardness}
    \end{center}
\end{figure}
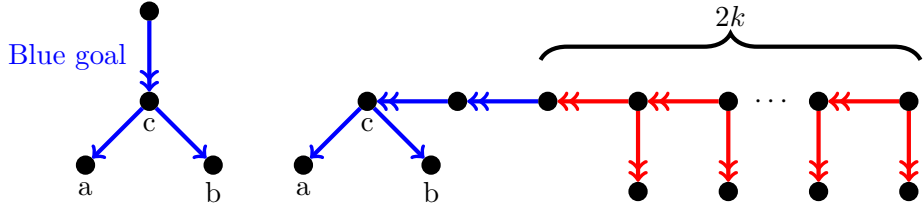

Although the use of these two additional vertices disappointingly forces future reductions to include two extra gadgets, these gadgets are hopefully relatively easy to concoct.  The Blue-to-Red vertex is very similar to the Variable vertex
and the Red-Or is literally the same thing as Or, but with Red edges instead of Blue.

% I think this would be overdoing it.
%TODO: should we add a section to summarize what is needed to use each version of Constraint Logic that we've defined?

\section{General Graph Reduction Details}
\label{sec:reduction}

Here, we show the gadgets used in game reduction.  We reduce from standard \rsBCL{}, so we need to prove the following items:

\begin{itemize}
    \item There exist gadgets that act like each of the Variable, And, Or, Split, and Choice vertices.
    \item We need a Goal gadget that Blue wins with if they can activate it while the rest of the construction remains consistent.
    \item The gadgets need to use a common interface between them that appropriately propagate active and inactive signals.
    \item Red has some way to make their $k$ extra moves while Blue is making moves on the circuitry.
\end{itemize}

We consider the game on general graphs, as well as when restricted to the Cartesian grid and triangular
grid.  For those stricter results, we will have some additional requirements:

\begin{itemize}
    \item The gadgets must ``snap'' onto the grid structures, meaning the graphs only include a subset of those grids in each case.
    \item We need additional wire gadgets to connect the interfaces between gadgets.
\end{itemize}

First, in Figure \ref{fig:input} we show the gadget (input/output) interface, on both the triangular and Cartesian  grids.  The I stands for Inactive; the A stands for Active.
\begin{figure}[ht]
    \centering
     \begin{tikzpicture} [every node/.style={draw,circle,inner sep=1pt,minimum size=1pt, very thick,fill=black}]
        \node (a) at (-3,0){};
        \node (b) at (-4.5,0){};
        \node (c) at (-1.5,0){};
        \node (d) at (-3,1.5){};
        \node (r) at (-3.75,-.75){};
        \node (s) at (-2.25, -.75){};
        \draw[blue] (a) edge node[below=.5mm, draw=none, fill=none, text=black]{I} (b) edge node[below, draw=none, fill=none, text=black]{A} (c) edge (d);
        \draw[red] (r) edge (s);
        \node (aa) at (3,0){};
        \node (bb) at (1.5,0){};
        \node (cc) at (4.5,0){};
        \node (dd) at (2.25,1){};
        \node (rr) at (3.75,-.75){};
        \node (ss) at (2.25, -.75){};
        \draw[blue] (aa) edge node[below=.5mm, draw=none, fill=none, text=black]{I} (bb) edge node[below, draw=none, fill=none, text=black]{A} (cc) edge (dd);
        \draw[red] (rr) edge (ss);
    \end{tikzpicture}
    \caption{INPUT/OUTPUT interface for Cartesian and Triangular grids}
     \label{fig:input}
\end{figure}
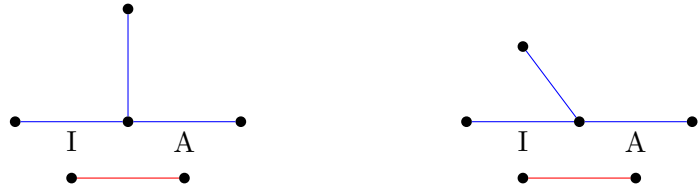

Of the four vertices in the interface gadget, the top and middle-bottom will never be incident to any other edges and I and A will never be connected to any red edges.  Thus, it never benefits Blue to choose the top edge over I or A, as choosing either the I or A edge instead may remove more blue edges and will not remove any extra red edges.  If Blue plays on the A edge of any interface, that will correlate to an active signal between those gadgets; if they play on the I edge, this simulates an inactive signal.

In \rsBCL{}, Blue never prefers inactive connections, so we design the gadgets so that Blue will choose to play on A edges (instead of I) any time it will not cost them an additional move.  The red edge below is included so Red has their own move to make while Blue plays on the gadget.  %Regardless of whether Blue takes Active or Inactive, Red takes a corresponding red edge, below the input/output.

In Figure \ref{fig:goal} we show the singular goal gadget, which will Blue to win the game if they are able to activate its input and not make any extra moves on any gadgets.

\begin{figure}[ht]
    \centering
     \begin{tikzpicture} [every node/.style={draw,circle,inner sep=1pt,minimum size=1pt, very thick,fill=black}]
        \node (a) at (-3,0){};
        \node (b) at (-4.5,0){};
        \node (c) at (-1.5,0){};
        \node (d) at (-3,1.5){};
        \node (e) at (0,0) {};
        \node (r) at (-3.75,-.75){};
        \node (s) at (-2.25, -.75){};
        \draw[blue]
            (a) edge node[below=.5mm, draw=none, fill=none, text=black]{I} (b) edge node[below, draw=none, fill=none, text=black]{A} (c) edge (d)
            (c) edge node[below, draw=none, fill=none, text=black]{G} (e);
        \draw[red]
            (r) edge (s);
    \end{tikzpicture}
    \caption{Goal gadget.  If the input is activated, then Blue doesn't need to spend an extra turn to take the extra edge, $G$.}
     \label{fig:goal}
\end{figure}
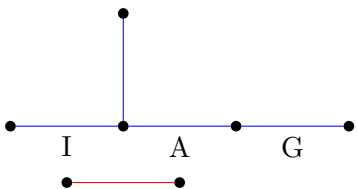

We show the beginning variables in the Cartesian and triangular grids in Figure \ref{fig:variable}. % and \ref{fig:tri-variable}.
Identical logic follows in showing that both are positive formulas.

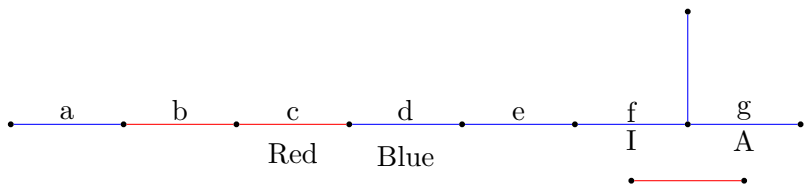
\begin{figure}
    \centering
    \begin{tikzpicture}[every node/.style={draw,circle,inner sep=0pt,minimum size=1pt, very thick,fill=black}]
    \node (a) at (-6,0){};
    \node (b) at (-4.5,0){};
    \node (c) at (-3,0){};
    \node (d) at (-1.5,0){};
    \node (e) at (0,0){};
    \node (f) at (1.5,0){};
    \node (g) at (3,0){};
    \node (h) at (4.5,0){};
    %\node (k) at (6,0){};
    \node (j) at (3.0,1.5){};
    \node (r) at (2.25,-.75){};
    \node (s) at (3.75,-.75){};
    \draw[blue]
    (a) edge node[above, draw=none, fill=none, text=black] {a} (b)
    (d) edge node[above, draw=none, fill=none, text=black] {d} node [below, draw=none, fill=none, text=black] {Blue} (e)
    (e) edge node[above, draw=none, fill=none, text=black] {e} (f)
    (f) edge node[above, draw=none, fill=none, text=black] {f} node[below=.5mm, draw=none, fill=none, text=black] {I} (g)
    (g) edge node[above, draw=none, fill=none, text=black] {g} node[below, draw=none, fill=none, text=black] {A} (h);
    %(h) edge node[above, draw=none, fill=none, text=black] {i} (k);
    \draw[red]
    (b) edge node[above, draw=none, fill=none, text=black] {b} (c)
    (c) edge node[above, draw=none, fill=none, text=black] {c} node[below, draw=none, fill=none, text=black] {Red}  (d);
    \draw[blue] (g) edge (j);
    \draw[red] (r) edge (s);
    \end{tikzpicture}
    \caption{VARIABLE gadget on Cartesian grids.  The triangular grid version is the same, but with the vertical edge in the output gadget bent to either the left or right.  If Red goes first, their best move is to play on edge c.  If Blue goes first, they should choose to play on d.  In addition to all these edges, each pair of Variable gadgets is accompanied by a single Red edge, different from the one pictured.} %works for both Cartesian and triangular grids
    \label{fig:variable}
\end{figure}

In the Variable gadget:
\begin{itemize}
    \item If Red goes first, they should choose to play on edge $c$.  This is because it leaves $a$ as a completely independent play that Blue has to make later.  Since Red will play on $c$ in half of these gadgets, these singleton Blue edges will exactly balance out the added Red edges.  Additionally, edge $e$ will remain, so Blue will have to continue on the gadget by playing $f$ instead of $g$.  Playing on $b$ is not as good for Red because Blue will (have to) play on $e$ and $g$.  Although that includes the extra move, it also means they get to activate the output.
    \item If Blue goes first, they should choose to play on edge $d$.  This allows them to activate the output by choosing $g$ in a future turn.  For the separated path $ab$, both players want to avoid this component and play to force the other player to move there instead.  There will be enough other Blue moves that these will constitute the last of Red's turns.  If Blue played on $e$ instead of $d$, then Red could take the $c$ pieces instead, forcing Blue to take more turns on the $a$ edges.
    % ↓↓↓ Kanae wrote an explanation. Is this clear?
    %\textcolor{red}{As a result, Red is forced to choose $b$ and clear $a$; Under optimal play both players make the same number of moves within all the gadgets described below. After Blue makes two moves ($d$ and $g$) and Red makes one move (the isolated red edge), the player to move is Red and they have to choose $b$. Although Blue can also choose $g$ by choosing $e$, it is worse for Blue; As in the discussion above, Red is forced to play on either $b$ or $c$. If Red chooses $c$, $a$ is not cleared and Blue has an additional move.}
\end{itemize}

%It is also important that we are able to ensure that when we connect variables they line up. This is always possible
If we have wires of even and odd length, then we can always arrange them to connect our gadgets on both of the grids that we're targeting, as we can include them until the variables align.  We get these wires in Figures \ref{fig:EWire} (even) and \ref{fig:OWire} (odd).

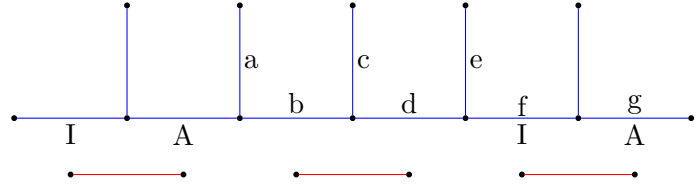
\begin{figure}[ht]
    \centering
    \begin{tikzpicture}[every node/.style={draw,circle,inner sep=0pt,minimum size=1pt, very thick,fill=black}]
    \node (a) at (0,0){};
    \node (b) at (1.5,0){};
    \node (c) at (3,0){};
    \node (d) at (4.5,0){};
    \node (e) at (6,0){};
    \node (f) at (7.5,0){};
    \node (g) at (9,0){};
    \node (h) at (1.5,1.5){};
    \node (i) at (3,1.5){};
    \node (j) at (4.5,1.5){};
    \node (k) at (6,1.5){};
    \node (l) at (7.5,1.5){};
     \node (r) at (.75,-.75){};
    \node (s) at (2.25,-.75){};
    \node (rr) at (3.75,-.75){};
    \node (ss) at (5.25,-.75){};
    \node (rrr) at (6.75,-.75){};
    \node (sss) at (8.25,-.75){};
    \draw[red] (r) edge (s);
    \draw[red] (rr) edge (ss);
    \draw[red] (rrr) edge (sss);
    \draw[blue] (b) edge node[below=.5mm, draw=none, fill=none, text=black] {I} (a) edge node[below, draw=none, fill=none, text=black] {A} (c) edge (h);
    \draw[blue] (c) edge node[right, draw=none, fill=none, text=black] {a}(i) edge node[above, draw=none, fill=none, text=black] {b} (d);
    \draw[blue] (d) edge node[right, draw=none, fill=none, text=black] {c} (j) edge node[above, draw=none, fill=none, text=black] {d}(e);
    \draw[blue] (e) edge node[right, draw=none, fill=none, text=black] {e} (k) edge node[above, draw=none, fill=none, text=black] {f} node[below=.5mm, draw=none, fill=none, text=black] {I} (f);
    \draw[blue] (f) edge (l) edge node[above, draw=none, fill=none, text=black] {g} node[below, draw=none, fill=none, text=black] {A} (g);
    \end{tikzpicture}
    \caption{Wire of Even Length.}
     \label{fig:EWire}
\end{figure}

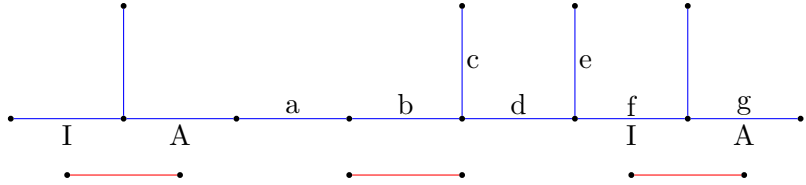
\begin{figure}[ht]
    \centering
    \begin{tikzpicture}[every node/.style={draw,circle,inner sep=0pt,minimum size=1pt, very thick,fill=black}]
    \node (a) at (0,0){};
    \node (b) at (1.5,0){};
    \node (c) at (3,0){};
    \node (d) at (4.5,0){};
    \node (e) at (6,0){};
    \node (f) at (7.5,0){};
    \node (g) at (9,0){};
    \node (h) at (10.5,0){};
    \node (i) at (1.5,1.5){};
    \node (j) at (6,1.5){};
    \node (k) at (7.5,1.5){};
    \node (l) at (9,1.5){};
    \node (r) at (.75,-.75){};
    \node (s) at (2.25,-.75){};
    \node (rr) at (4.5,-.75){};
    \node (ss) at (6,-.75){};
    \node (rrr) at (9.75,-.75){};
    \node (sss) at (8.25,-.75){};
    \draw[red] (r) edge (s);
    \draw[red] (rr) edge (ss);
     \draw[red] (rrr) edge (sss);
    \draw[blue] (b) edge node[below=.5mm, draw=none, fill=none, text=black] {I} (a)
        edge node[below, draw=none, fill=none, text=black] {A} (c) edge (i);
    \draw[blue] (d) edge node[above, draw=none, fill=none, text=black] {a} (c)
        edge node[above, draw=none, fill=none, text=black] {b} (e);
    \draw[blue] (e) edge node[above, draw=none, fill=none, text=black] {d} (f)
        edge node[right, draw=none, fill=none, text=black] {c}(j);
    \draw[blue] (f) edge node[right, draw=none, fill=none, text=black] {e}(k);
    \draw[blue] (g) edge node[above, draw=none, fill=none, text=black] {f} node[below=.5mm, draw=none, fill=none, text=black] {I} (f)
        edge node[above, draw=none, fill=none, text=black] {g} node[below, draw=none, fill=none, text=black] {A} (h) edge (l);
    \end{tikzpicture}
    \caption{Wire of Odd Length.}
     \label{fig:OWire}
\end{figure}

The argument is identical for the even and odd cases. With Active input, $a$ is
cleared so Blue selects $d$ which then allows them to select $g$ as the
output. With Inactive input, Blue must clear $a$, so they first select $b$. Then,
they must still clear $e$, so they select $f$, leading to an Inactive output.

\subsection{General Graph Gadgets}

We now outline the reduction gadgets used on a general graph. First, consider the AND gadget in Figure \ref{fig:and}.  In this gadget, Blue can only activate the output gadget if both input gadgets are active without playing additional moves.

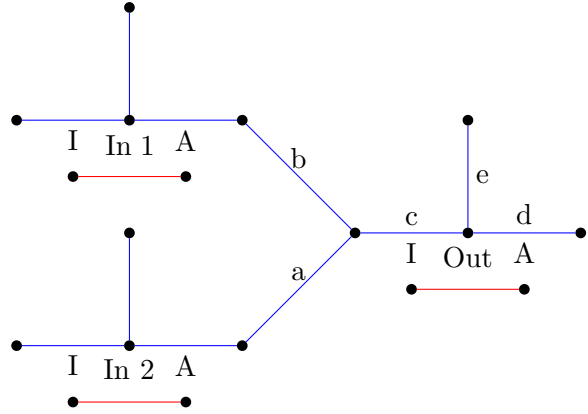
\begin{figure}[!ht]
    \centering
    \begin{tikzpicture}[every node/.style={draw,circle,inner sep=1pt,minimum size=1pt, very thick,fill=black}]
    \node (a) at (0,0){};
    \node (b) at (1.5,0){};
    \node (c) at (3,0){};
    \node (d) at (1.5,1.5){};
    \node (e) at (0,3){};
    \node (f) at (1.5,3){};
    \node (g) at (3,3){};
    \node (h) at (1.5,4.5){};
    \node (i) at (4.5,1.5){};
    \node (j) at (6,1.5){};
    \node (k) at (7.5,1.5){};
    \node (l) at (6,3){};
    \node (r) at (.75,-.75){};
    \node (s) at (2.25,-.75){};
    \node (rr) at (.75,2.25){};
    \node (ss) at (2.25,2.25){};
    \node (rrr) at (5.25,.75){};
    \node (sss) at (6.75,.75){};
    \draw[red] (r) edge node [above, draw=none, fill=none, text=black] {In 2} (s);
    \draw[red] (rr) edge node [above, draw=none, fill=none, text=black] {In 1} (ss);
    \draw[red] (rrr) edge node [above, draw=none, fill=none, text=black] {Out} (sss);
    \draw[blue] (b) edge node[below=.5mm, draw=none, fill=none, text=black] {I} (a)
        edge (d) edge node[below, draw=none, fill=none, text=black] {A} (c);
    \draw[blue] (i) edge node[above, draw=none, fill=none, text=black] {a} (c)
        edge node[above, draw=none, fill=none, text=black] {b} (g) edge
        node[below=.5mm, draw=none, fill=none, text=black] {I} node[above,
        draw=none, fill=none, text=black] {c} (j);
    \draw[blue] (f) edge node[below=.5mm, draw=none, fill=none, text=black] {I} (e)
        edge node[below, draw=none, fill=none, text=black] {A} (g) edge (h);
    \draw[blue] (j) edge node[above, draw=none, fill=none, text=black] {d}
        node[below, draw=none, fill=none, text=black] {A} (k) edge node[right,
        draw=none, fill=none, text=black] {e} (l);
    \end{tikzpicture}
    \caption{AND gadget.  The output (Out) can be activated only if both inputs (In 1 and In 2) are active.}
    \label{fig:and}
\end{figure}

    \begin{enumerate}
        \item If the input is $(A,A)$, then both edge $a$ and $b$ vanish,
            leaving only edges $c,d,e$. Blue then takes $d$,
            giving an active output;
        \item If the input is $(I,A)$, then the edges $b,c,d,e$ remain. Blue then would take $c$, giving an inactive output;
        \item If the input is $(A,I)$, the output is inactive due to similar
            reasoning as in Case 2;
        \item If the input is $(I,I)$, then edges $a,b,c,d,e$ will remain, in
            which case Blue can take them all by claiming edge $c$, giving an
            inactive output.
    \end{enumerate}

Next, consider the OR gadget in Figure \ref{fig:or}, with two inputs and one output.  The output can be activated if either of the two inputs is active, but if both are inactive then the output must also be inactive.  Blue will make an extra move on the gadget outside of the three interfaces, so an additional red edge is included.

\begin{figure}[ht]
    \centering
    \begin{tikzpicture}[every node/.style={draw,circle,inner sep=1pt,minimum
        size=1pt, very thick,fill=black}]
    \node (a) at (0,0){};
    \node (b) at (1.5,0){};
    \node (c) at (3,0){};
    \node (d) at (1.5,1.5){};
    \node (e) at (0,3){};
    \node (f) at (1.5,3){};
    \node (g) at (3,3){};
    \node (h) at (1.5,4.5){};
    \node (i) at (4.5,1.5){};
    \node (j) at (7.5,1.5){};
    \node (k) at (9,1.5){};
    \node (l) at (6,1.5){};
    \node (m) at (10.5,1.5){};
    \node (o) at (9,3){};
    \node (r) at (.75,-.75){};
    \node (s) at (2.25,-.75){};
    \node (rr) at (.75,2.25){};
    \node (ss) at (2.25,2.25){};
    \node (rrrr) at (6,.75){};
    \node (ssss) at (4.5,.75){};
    \node (rrr) at (9.75,.75){};
    \node (sss) at (8.25,.75){};
    \draw[red] (r) edge (s);
    \draw[red] (rr) edge (ss);
    \draw[red] (rrr) edge (sss);
    \draw[red] (rrrr) edge (ssss);
    \draw[blue] (f) edge node[below=.5mm, draw=none, fill=none, text=black] {I} (e)
        edge node[below, draw=none, fill=none, text=black] {A} (g) edge (h);
    \draw[blue] (b) edge node[below=.5mm, draw=none, fill=none, text=black] {I} (a)
        edge node[below, draw=none, fill=none, text=black] {A} (c) edge (d);
    \draw[blue] (g) edge node[right, draw=none, fill=none, text=black]{a} (c)
        edge node[above, draw=none, fill=none, text=black] {b} (i);
    \draw[blue] (i) edge node[above, draw=none, fill=none, text=black] {c} (c) edge node[above, draw=none, fill=none, text=black]{d} (l);
    \draw[blue] (l) edge node[above, draw=none, fill=none, text=black]{e} (j);
    \draw[blue] (k) edge node[above, draw=none, fill=none, text=black] {f} node[below=.5mm, draw=none, fill=none, text=black] {I} (j)
        edge node[above, draw=none, fill=none, text=black] {g} node[below, draw=none, fill=none, text=black] {A} (m) edge (o);
    \end{tikzpicture}
    \caption{OR gadget.}
     \label{fig:or}
\end{figure}
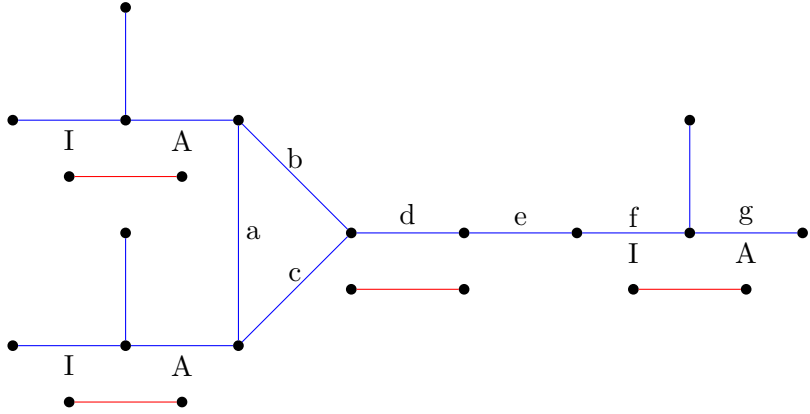
  \begin{enumerate}
         \item If the input is $(A,A)$, edges $a,b$ and $c$ vanish. Blue should take $e$ in order to clear $d$. They then select $g$, giving an active output;
        \item If the input is $(I,A)$, $a$ and $c$ vanish. Blue would then select $d$, which clears $b$ and $e$. They then select $g$, resulting in an active output;
        \item If the input is $(A,I)$, the output is active due to similar reasoning in Case 2;
        \item If the input is $(I,I)$, then edges $a,b,$ and $c$ remain. Blue would then select $b$ or $c$, which clears $a,b,c,$ and $d$. In order to clear $e$, Blue must choose $f$, resulting in an inactive output.
    \end{enumerate}
Following this, we consider the FANOUT gadget in Figure \ref{fig:fanout}:
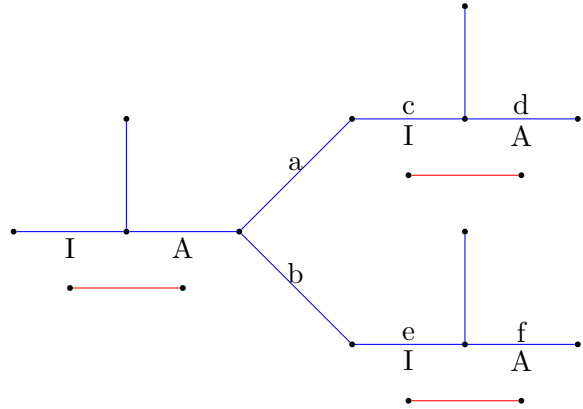
\begin{figure}[ht]
    \centering
    \begin{tikzpicture}[every node/.style={draw,circle,inner sep=0pt,minimum
        size=1pt, very thick,fill=black}]
    \node (a) at (0,1.5){};
    \node (b) at (1.5,1.5){};
    \node (c) at (3,1.5){};
    \node (d) at (1.5,3){};
    \node (e) at (4.5,0){};
    \node (f) at (6,0){};
    \node (g) at (7.5,0){};
    \node (h) at (6,1.5){};
    \node (i) at (4.5,3){};
    \node (j) at (6,3){};
    \node (k) at (7.5,3){};
    \node (l) at (6,4.5){};
    \node (r) at (.75,.75){};
    \node (s) at (2.25,.75){};
    \node (rr) at (6.75,2.25){};
    \node (ss) at (5.25,2.25){};
    \node (rrr) at (5.25,-.75){};
    \node (sss) at (6.75,-.75){};
    \draw[red] (r) edge (s);
    \draw[red] (rr) edge (ss);
    \draw[red] (rrr) edge (sss);
    \draw[blue] (b) edge node[below=.5mm, draw=none, fill=none, text=black] {I} (a)
        edge node[below, draw=none, fill=none, text=black] {A} (c) edge (d);
    \draw[blue] (c) edge node[above, draw=none, fill=none, text=black] {a} (i)
        edge node[above, draw=none, fill=none, text=black] {b}(e);
    \draw[blue] (j) edge node[above, draw=none, fill=none, text=black] {c} node[below=.5mm, draw=none, fill=none, text=black] {I} (i)
        edge node[above, draw=none, fill=none, text=black] {d} node[below, draw=none, fill=none, text=black] {A} (k) edge (l);
    \draw[blue] (f) edge node[above, draw=none, fill=none, text=black] {e} node[below=.5mm, draw=none, fill=none, text=black] {I} (e)
        edge node[above, draw=none, fill=none, text=black] {f} node[below, draw=none, fill=none, text=black] {A} (g) edge (h);
    \end{tikzpicture}
    \caption{FANOUT gadget.}
        \label{fig:fanout}
\end{figure}
  \begin{enumerate}
        \item If the input is $A$, then both edge $a$ and $b$ vanish,
            allowing Blue to take both $d$ and $f$, resulting in two active outputs;
        \item If the input is $I$, then the edges $a$ and $b$ remain,
            forcing Blue to select $c$ and $e$, leading to an inactive output for both branches.
    \end{enumerate}

To conclude our analysis of the general case, we consider the CHOICE gadget in Figure \ref{fig:choice}.

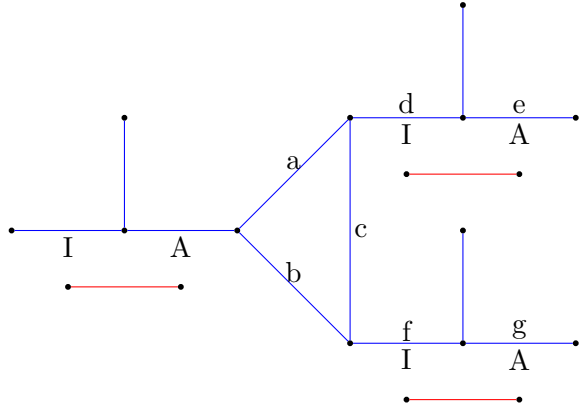
\begin{figure}
    \centering
    \begin{tikzpicture}[every node/.style={draw,circle,inner sep=0pt,minimum size=1pt, very thick,fill=black}]
    \node (a) at (0,1.5){};
    \node (b) at (1.5,1.5){};
    \node (c) at (3,1.5){};
    \node (d) at (1.5,3){};
    \node (e) at (4.5,0){};
    \node (f) at (6,0){};
    \node (g) at (7.5,0){};
    \node (h) at (6,1.5){};
    \node (i) at (4.5,3){};
    \node (j) at (6,3){};
    \node (k) at (7.5,3){};
    \node (l) at (6,4.5){};
    \node (r) at (.75,.75){};
    \node (s) at (2.25,.75){};
    \node (rr) at (6.75,2.25){};
    \node (ss) at (5.25,2.25){};
    \node (rrr) at (5.25,-.75){};
    \node (sss) at (6.75,-.75){};
    \draw[red] (r) edge (s);
    \draw[red] (rr) edge (ss);
    \draw[red] (rrr) edge (sss);
    \draw[blue] (b) edge node[below=.5mm, draw=none, fill=none, text=black] {I} (a)
        edge node[below, draw=none, fill=none, text=black] {A} (c) edge (d);
    \draw[blue] (c) edge node[above, draw=none, fill=none, text=black] {a} (i)
        edge node[above, draw=none, fill=none, text=black] {b} (e);
    \draw[blue] (j) edge node[above, draw=none, fill=none, text=black] {d} node[below=.5mm, draw=none, fill=none, text=black] {I} (i)
        edge node[above, draw=none, fill=none, text=black] {e} node[below, draw=none, fill=none, text=black] {A} (k) edge (l);
    \draw[blue] (f) edge node[above, draw=none, fill=none, text=black] {f} node[below=.5mm, draw=none, fill=none, text=black] {I} (e)
        edge node[above, draw=none, fill=none, text=black] {g} node[below, draw=none, fill=none, text=black] {A} (g) edge (h);
    \draw[blue] (e) edge node[right, draw=none, fill=none, text=black] {c} (i);
    \end{tikzpicture}
    \caption{CHOICE gadget.}
    \label{fig:choice}
\end{figure}
    \begin{enumerate}
        \item If the input is $A$, then both edge $a$ and $b$ vanish. Blue
            still needs to remove edge $c$, so they select inactive as one
            output, and active as the other;
        \item If the input is $I$, then the edges $a$ and $b$ remain,
            forcing Blue to take an inactive output for both branches.
    \end{enumerate}

\begin{theorem}
    \rsMisParck{} is \cclass{PSPACE}-complete.
\end{theorem}
\begin{proof}
    By the reduction from \rsBCL{} with the gadgets described above, \rsMisParck{} is \cclass{PSPACE}-hard.  \rsMisParck{} is also in \cclass{PSPACE} because the game tree's height is bounded by the number of edges and each position has at most that many options.  Thus, \rsMisParck{} is \cclass{PSPACE}-complete.
\end{proof}

Some of the gadgets from above need to be modified slightly in order to fit on either the triangular or Cartesian grids.  We describe these gadgets in the appendix in Section \ref{sec:specificGadgets}.

\section{Future Directions}

This work is the first result the authors know of that uses Constraint Logic show that mis\`ere combinatorial games are \cclass{PSPACE}-hard.  Naturally, this begs the question of whether this approach can be used for other mis\`ere games as well.  Additionally, although we did not use the mis\`ere version of Constraint Logic we devised (\rsMPBCL{}) in our reduction, it's also possible that that new ruleset could be used for future hardness results as well.  That would be even more useful if the number of basis vertices needed for hardness could be decreased.

\begin{problem}
    Is there a set of basis vertices with fewer than seven elements that \rsMPBCL{} is \cclass{PSPACE}-complete on, even for planar graphs?
\end{problem}

Of course, the big questions that motivated this are still unsolved.  Recent advances in the normal-play version of \ruleset{PArcK} have shown that the game is computationally intractable, specifically \cclass{NP}-hard \cite{HANAKA2026103716}.  It remains unknown whether the game is \cclass{PSPACE}-complete.

\begin{problem}
    What is the computational complexity of \ruleset{Partizan Arc Kayles}?  What is it's complexity on planar graphs?
\end{problem}

This, in turn, was motivated by two other games that have unknown complexity, despite plenty of interest from combinatorial game researchers.

\begin{problem}
    What is the computational complexity of \ruleset{Arc Kayles}?
\end{problem}

\begin{problem}
    What is the computational complexity of \ruleset{Domineering}?
\end{problem}

Although our work does not appear to yield a new path to proving any of these three, we hope it will ignite more interest in these long-standing problems.

\section*{Acknowledgements}

The authors would like to thank the organizers of the Integers Conference 2025,
held in Athens, Georgia, for bringing them together and facilitating this
research.
This work was supported by JSPS KAKENHI Grant Number JP24KJ1232.

\bibliographystyle{plainurl}

%%%%%%%%%%%%%%%%%%%%%%%%%%%%%%%%%%%%%%%%%%%%%%%%%%%%%%%%%%%%%%%%%%%%%%%%%%%%%%%%%%%%%%%%%%%%%%%%%%%%%%%%%%%%%%%%%%%%%%%%%%
\appendix

\section{Specific Gadgets for Cartesian and Triangular Grids}
\label{sec:specificGadgets}

In this section, we provide specific gadgets for Cartesian and Triangular grids when those need to differ from the versions shown previously.

\subsection{Cartesian grid gadgets}

Note first that the AND and FANOUT gadgets can be redrawn in the Cartesian grid
with no fundamental edits, thus the logic from the general case follows here. We show them in Figures \ref{fig:cart-and} and \ref{fig:cart-fanout} respectively.

\begin{figure}[h!]
        \centering
        \begin{minipage}{.45\textwidth}
        \centering
        \begin{tikzpicture}[every node/.style={draw,circle,inner sep=0pt,minimum size=1pt, very thick,fill=black}]
        \node (a) at (1,2){};
        \node (b) at (0,2){};
        \node (c) at (2,2){};
        \node (d) at (1,3){};
        \node (r) at (.5,1.5){};
        \node (s) at (1.5,1.5){};
        \node (aa) at (1,0){};
        \node (bb) at (0,0){};
        \node (cc) at (2,0){};
        \node (dd) at (1,1){};
        \node (rr) at (.5,-.5){};
        \node (ss) at (1,-.5){};
        \node (aaa) at (3,1){};
        \node (bbb) at (2,1){};
        \node (ccc) at (4,1){};
        \node (ddd) at (3,2){};
        \node (rrr) at (2.5,.5){};
        \node (sss) at (3.5,.5){};
        \draw[blue] (a) edge node[below=.5mm, draw=none, fill=none, text=black]{I}
            (b) edge node[below, draw=none, fill=none, text=black]{A} (c) edge
            (d);
        \draw[blue] (aa) edge node[below=.5mm, draw=none, fill=none, text=black]{I}
            (bb) edge node[below, draw=none, fill=none, text=black]{A} (cc)
            edge (dd);
        \draw[blue] (aaa) edge node[below=.5mm, draw=none, fill=none, text=black]{I}
            (bbb) edge node[below, draw=none, fill=none, text=black]{A} (ccc)
            edge (ddd);
        \draw[blue] (bbb) edge (c) edge (cc);
        \draw[red] (r) edge (s);
        \draw[red] (rr) edge (ss);
        \draw[red] (rrr) edge (sss);
  \end{tikzpicture}
  \caption{AND gadget for Cartesian grids.}
   \label{fig:cart-and}
        \end{minipage}\hfill
        \begin{minipage}{.45\textwidth}
    \centering
        \begin{tikzpicture}[every node/.style={draw,circle,inner sep=0pt,minimum
    size=1pt, very thick,fill=black}]
        \node (a) at (1,1){};
        \node (b) at (0,1){};
        \node (c) at (2,1){};
        \node (d) at (1,2){};
        \node (aa) at (3,0){};
        \node (bb) at (2,0){};
        \node (cc) at (4,0){};
        \node (dd) at (3,1){};
        \node (aaa) at (3,2){};
        \node (bbb) at (2,2){};
        \node (ccc) at (4,2){};
        \node (ddd) at (3,3){};
        \node (r) at (.5,.5){};
        \node (s) at (1.5,.5){};
         \node (rr) at (2.5,1.5){};
        \node (ss) at (3.5,1.5){};
         \node (rrr) at (2.5,-.5){};
        \node (sss) at (3.5,-.5){};
        \draw[blue] (a) edge node[below=.5mm, draw=none, fill=none, text=black]{I}
            (b) edge node[below, draw=none, fill=none, text=black]{A} (c) edge
            (d);
        \draw[blue] (aa) edge node[below=.5mm, draw=none, fill=none, text=black]{I}
            (bb) edge node[below, draw=none, fill=none, text=black]{A} (cc)
            edge (dd);
        \draw[blue] (aaa) edge node[below=.5mm, draw=none, fill=none, text=black]{I}
            (bbb) edge node[below, draw=none, fill=none, text=black]{A} (ccc)
            edge (ddd);
        \draw[blue] (c) edge (bb) edge (bbb);
        \draw[red] (r) edge (s);
        \draw[red] (rr) edge (ss);
        \draw[red] (rrr) edge (sss);
  \end{tikzpicture}
  \caption{FANOUT gadget for Cartesian grids.}
    \label{fig:cart-fanout}
        \end{minipage}
 \end{figure}
Since our gadgets for our OR and CHOICE can not be immediately redrawn in the
Cartesian grid, we give new gadgets for the Cartesian grid below. First consider the OR gadget in Figure \ref{fig:cart-or}.

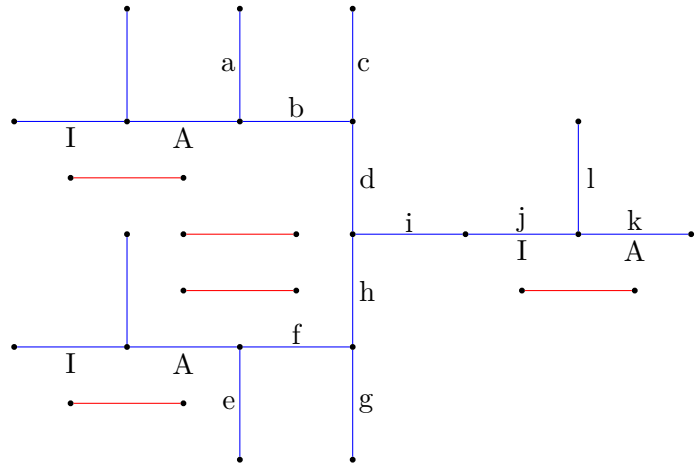
\begin{figure}[ht]
\centering
\begin{tikzpicture}[every node/.style={draw,circle,inner sep=0pt,minimum
    size=1pt, very thick,fill=black}]
        \node (a) at (0,0){};
        \node (b) at (1.5,0){};
        \node (c) at (3,0){};
        \node (d) at (1.5,1.5){};
        \node (e) at (0,3){};
        \node (f) at (1.5,3){};
        \node (h) at (1.5,4.5){};
        \node (g) at (3,3){};
        \node (i) at (4.5,0){};
        \node (j) at (3,4.5){};
        \node (k) at (4.5,1.5){};
        \node (l) at (4.5,3){};
        \node (m) at (6,1.5){};
        \node (n) at (7.5,1.5){};
        \node (o) at (9,1.5){};
        \node (p) at (3,-1.5){};
        \node (q) at (4.5,-1.5){};
        \node (u) at (4.5,4.5){};
        \node (v) at (7.5,3){};
        \node (r) at (.75,-.75){};
        \node (s) at (2.25,-.75){};
        \node (rr) at (.75,2.25){};
        \node (ss) at (2.25,2.25){};
        \node (rrr) at (6.75,.75){};
        \node (sss) at (8.25,.75){};
        \node (rrrr) at (2.25,.75){};
        \node (ssss) at (3.75,.75){};
        \node (rrrrr) at (2.25,1.5){};
        \node (sssss) at (3.75,1.5){};
        \draw[red] (r) edge (s);
        \draw[red] (rr) edge (ss);
        \draw[red] (rrr) edge (sss);
        \draw[red] (rrrr) edge (ssss);
        \draw[red] (rrrrr) edge (sssss);
        \draw[blue] (b) edge node[below=.5mm, draw=none, fill=none, text=black]{I}
            (a) edge node[below, draw=none, fill=none, text=black]{A} (c) edge
            (d);
        \draw[blue] (f) edge node[below=.5mm, draw=none, fill=none, text=black]{I}
            (e) edge node[below, draw=none, fill=none, text=black]{A} (g) edge
            (h);
        \draw[blue] (l) edge node[above, draw=none, fill=none, text=black]{b}
            (g) edge node[right, draw=none, fill=none, text=black]{d} (k);
        \draw[blue] (i) edge node[above, draw=none, fill=none, text=black]{f}
            (c) edge node[right, draw=none, fill=none, text=black]{h} (k);
        \draw[blue] (m) edge node[above, draw=none, fill=none, text=black]{i}
            (k) edge node[below=.5mm, draw=none, fill=none, text=black]{I}
            node[above, draw=none, fill=none, text=black]{j} (n);
        \draw[blue] (j) edge node[left, draw=none, fill=none, text=black]{a} (g);
        \draw[blue] (u) edge node[right, draw=none, fill=none, text=black]{c} (l);
        \draw[blue] (c) edge node[left, draw=none, fill=none, text=black]{e} (p);
        \draw[blue] (i) edge node[right, draw=none, fill=none, text=black]{g} (q);
        \draw[blue] (n) edge node[right, draw=none, fill=none, text=black]{l}
            (v) edge node[below, draw=none, fill=none, text=black]{A}
            node[above, draw=none, fill=none, text=black]{k} (o);
  \end{tikzpicture}
  \caption{OR gadget for Cartesian grids.}
    \label{fig:cart-or}
\end{figure}
  \begin{enumerate}
      \item If the input is $(A,A)$, then edges $a,b,e,f$ are cleared upon
          input. Because Blue needs to clear $c$, Blue has to take edge $c$ or
          $d$, in which case they would prefer $d$ in an optimal strategy. Upon
          taking $d$, edges $c,d,i,h$ are cleared, leaving edges $g,j,l,k$.
          Then Blue would need to take $g$ and $k$, rendering an active output;
      \item If the input is $(A,I)$, then edges $a,b$ are cleared upon input.
          Because Blue needs to clear $c$, Blue has to take either $c$ or $d$,
          in which case they would prefer $d$ in an optimal strategy. Upon
taking $d$, the edges $c,d,h,i$ are cleared, leaving the edges $e,f,g,j,k,l$.
          Then by claiming $f$ and $k$, Blue clears them all, rendering an
          active output;
      \item If the input is $(I,A)$, the output is active due to similar reasoning in Case 2;
      \item If the input is $(I,I)$, then none of the labeled edges are cleared
          upon input. Because Blue has to clear edge $a$, Blue has to take either $a$ or $b$, in which case they would prefer $b$ in an optimal strategy, clearing $a,b,c,d$. Similarly, Blue would claim edge $f$ in an optimal strategy, clearing $e,f,g,h$. At this point, only
$i,j,k,l$ remains and Blue would take $j$ to clear them all,
          rendering an inactive output.
  \end{enumerate}
To conclude our gadget analysis for the game in the Cartesian grid, we show the CHOICE gadget in \ref{fig:cart-choice}.
\begin{figure}[ht]
\centering
\begin{tikzpicture}[every node/.style={draw,circle,inner sep=0pt,minimum
    size=1pt, very thick,fill=black}]
        \node (a) at (1.5,1.5){};
        \node (b) at (0,1.5){};
        \node (c) at (3,1.5){};
        \node (d) at (1.5,3){};
        \node (da) at (3,0){};
        \node (dc) at (3,3){};
        \node (ea) at (4.5,0){};
        \node (eb) at (4.5,1.5){};
        \node (ec) at (4.5,3){};
        \node (fa) at (6,0){};
        \node (fb) at (6,1.5){};
        \node (fc) at (6,3){};
        \node (ga) at (7.5,0){};
        \node (gc) at (7.5,3){};
        \node (aa) at (10.5,0){};
        \node (bb) at (9,0){};
        \node (cc) at (12,0){};
        \node (dd) at (10.5,1.5){};
        \node (aaa) at (10.5,3){};
        \node (bbb) at (9,3){};
        \node (ccc) at (12,3){};
        \node (ddd) at (10.5,4.5){};
        \node (r) at (.75,.75){};
        \node (s) at (2.25,.75){};
        \node (rr) at (9.75,2.25){};
        \node (ss) at (11.25,2.25){};
        \node (rrr) at (9.75,-.75){};
        \node (sss) at (11.25,-.75){};
        \node (rrrr) at (6.75,.75){};
        \node (ssss) at (8.25,.75){};
        \node (rrrrr) at (6.75,1.5){};
        \node (sssss) at (8.25,1.5){};
        \node (rrrrrr) at (6.75,2.25){};
        \node (ssssss) at (8.25,2.25){};
        \draw[blue] (a) edge node[below=.5mm, draw=none, fill=none, text=black]{I}
            (b) edge node[below, draw=none, fill=none, text=black]{A} (c) edge
            (d);
        \draw[blue] (aa) edge node[below=.5mm, draw=none, fill=none, text=black]{I}
            node[above, draw=none, fill=none, text=black]{p} (bb) edge node[below, draw=none, fill=none, text=black]{A} node[above, draw=none, fill=none, text=black]{r}  (cc)
            edge (dd);
        \draw[blue] (aaa) edge node[below=.5mm, draw=none, fill=none, text=black]{I}
            node[above, draw=none, fill=none, text=black]{o} (bbb) edge node[below, draw=none, fill=none, text=black]{A} node[above, draw=none, fill=none, text=black]{q}(ccc)
            edge (ddd);
        \draw[blue] (c) edge node[right, draw=none, fill=none, text=black] {b}
            (da) edge node[right, draw=none, fill=none, text=black] {a} (dc)
            edge node[above, draw=none, fill=none, text=black] {d} (eb);
        \draw[blue] (ea) edge node[above, draw=none, fill=none, text=black]
            {e}(da) edge node[above, draw=none, fill=none, text=black] {j} (fa)
            edge node[right, draw=none, fill=none, text=black] {g} (eb);
        \draw[blue] (ec) edge node[above, draw=none, fill=none, text=black] {c}
            (dc) edge node[above, draw=none, fill=none, text=black] {h} (fc)
            edge node[right, draw=none, fill=none, text=black] {f}(eb);
        \draw[blue] (eb) edge node[above, draw=none, fill=none, text=black]
            {i}(fb);
        \draw[blue] (gc) edge node[above, draw=none, fill=none, text=black] {k}
            (fc) edge node[above, draw=none, fill=none, text=black] {m} (bbb);
        \draw[blue] (ga) edge node[above, draw=none, fill=none, text=black] {l}
            (fa) edge node[above, draw=none, fill=none, text=black] {n} (bb);
        \draw[red] (r) edge (s);
        \draw[red] (rr) edge (ss);
        \draw[red] (rrr) edge (sss);
        \draw[red] (rrrr) edge (ssss);
        \draw[red] (rrrrr) edge (sssss);
        \draw[red] (rrrrrr) edge (ssssss);
  \end{tikzpicture}
  \caption{CHOICE gadget for Cartesian grids.}
    \label{fig:cart-choice}
\end{figure}
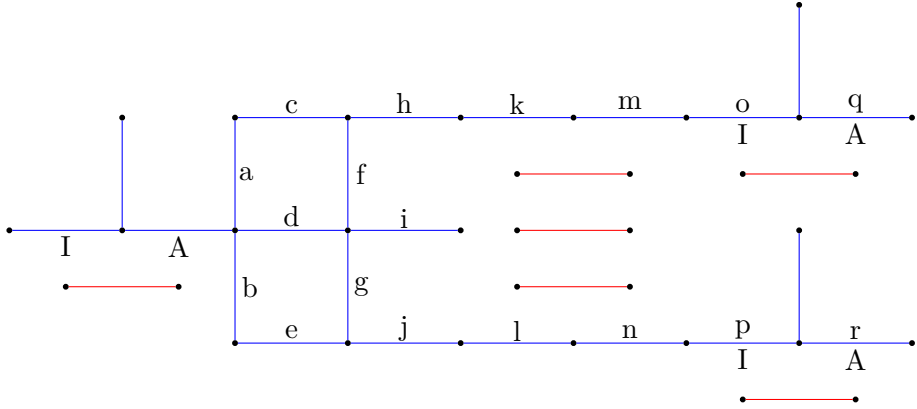
   \begin{enumerate}
        \item If the input is $A$, it clears $a,b$ and $d$. Blue must then
            choose $f$ or $g$ to remove $i$. Suppose they choose $f$, thereby
            clearing $c$ and $h$. They must then select $m$ in order to clear
            $k$, and following this they can choose $q$ on the top row, leading to an active output on top. On the
            bottom row, they still must clear $e$ so they select $j$ which then
            forces them to select $I$ in order to clear $n$. Symmetrical logic
            follows if they choose $g$;
        \item If the input is $I$, Blue should choose $d$ which clears
            $a,b,f,g$ and $i$. Note that the top and bottom rows are symmetric.
            Blue must choose $h$ and $j$ in order to clear, respectively, $c$
            and $e$. This forces them to choose $o$ to clear $m$ and $p$ to clear $n$, leading to an inactive output in both rows.
    \end{enumerate}

\subsection{Triangular grid gadgets}

Again, our gadgets for AND and FANOUT gadgets can be redrawn in the Cartesian
grid with no fundamental edits as shown in Figures \ref{fig:tri-and} and \ref{fig:tri-fanout}. Thus, the logic from the general case follows.

\begin{figure}[ht]
    \centering
    \begin{minipage}{.45\textwidth}
\centering
    \begin{tikzpicture}[every node/.style={draw,circle,inner sep=0pt,minimum size=1pt, very thick,fill=black}]
    \node (a) at (0,0){};
    \node (b) at (1,0){};
    \node (c) at (2,0){};
    \node (d) at (0.5,.75){};
    \node (e) at (0,2){};
    \node (f) at (1,2){};
    \node (g) at (2,2){};
    \node (h) at (0.5,2.75){};
    \node (i) at (3,1){};
    \node (j) at (4,1){};
    \node (k) at (5,1){};
    \node (l) at (3.5,1.75){};
     \node (r) at (.5,-.5){};
        \node (s) at (1.5,-.5){};
         \node (rr) at (.5,1.5){};
        \node (ss) at (1.5,1.5){};
         \node (rrr) at (3.5,.5){};
        \node (sss) at (4.5,.5){};
        \draw[red] (r) edge (s);
        \draw[red] (rr) edge (ss);
        \draw[red] (rrr) edge (sss);
    \draw[blue] (b) edge node[below=.5mm, draw=none, fill=none, text=black] {I} (a)
        edge (d) edge node[below, draw=none, fill=none, text=black] {A} (c);
    \draw[blue] (i) edge (c) edge (g) edge node[below=.5mm, draw=none, fill=none,
        text=black] {I} (j);
    \draw[blue] (f) edge node[below=.5mm, draw=none, fill=none, text=black] {I} (e)
        edge node[below, draw=none, fill=none, text=black] {A} (g) edge (h);
    \draw[blue] (j) edge node[below, draw=none, fill=none, text=black] {A} (k)
        edge (l);
    \end{tikzpicture}
    \caption{AND gadget for triangular grids.}
        \label{fig:tri-and}
    \end{minipage} \hfill
    \begin{minipage}{.45\textwidth}
    \centering
    \begin{tikzpicture}[every node/.style={draw,circle,inner sep=0pt,minimum
        size=1pt, very thick,fill=black}]
    \node (a) at (0,1){};
    \node (b) at (1,1){};
    \node (c) at (2,1){};
    \node (d) at (0.5,1.75){};
    \node (e) at (3,0){};
    \node (f) at (4,0){};
    \node (g) at (5,0){};
    \node (h) at (3.5,.75){};
    \node (i) at (3,2){};
    \node (j) at (4,2){};
    \node (k) at (5,2){};
    \node (l) at (3.5,2.7){};
     \node (r) at (.5,.5){};
        \node (s) at (1.5,.5){};
         \node (rr) at (3.5,1.5){};
        \node (ss) at (4.5,1.5){};
         \node (rrr) at (3.5,-.5){};
        \node (sss) at (4.5,-.5){};
        \draw[red] (r) edge (s);
        \draw[red] (rr) edge (ss);
        \draw[red] (rrr) edge (sss);
    \draw[blue] (b) edge node[below=.5mm, draw=none, fill=none, text=black] {I} (a)
        edge node[below, draw=none, fill=none, text=black] {A} (c) edge (d);
    \draw[blue] (c) edge (i) edge (e);
    \draw[blue] (j) edge node[below=.5mm, draw=none, fill=none, text=black] {I} (i)
        edge node[below, draw=none, fill=none, text=black] {A} (k) edge (l);
    \draw[blue] (f) edge node[below=.5mm, draw=none, fill=none, text=black] {I} (e)
        edge node[below, draw=none, fill=none, text=black] {A} (g) edge (h);
    \end{tikzpicture}
    \caption{FANOUT gadget for triangular grids.}
    \label{fig:tri-fanout}
    \end{minipage}
\end{figure}
Then, we give an alternate OR gadget for the triangular grid in Figure \ref{fig:tri-or}.
\begin{figure}
    \centering
    \begin{tikzpicture}[every node/.style={draw,circle,inner sep=0pt,minimum size=1pt, very thick,fill=black}]
    \node (a) at (0,0){};
    \node (b) at (1.5,0){};
    \node (c) at (3,0){};
    \node (d) at (0.75,1){};
    \node (e) at (0,2){};
    \node (f) at (1.5,2){};
    \node (g) at (3,2){};
    \node (h) at (0.75,3){};
    \node (i) at (3.75,1){};
    \node (j) at (5.25,1){};
    \node (k) at (6.75,1){};
    \node (l) at (4.5,2){};
    \node (m) at (8.25,1){};
    \node (n) at (9.75,1){};
    \node (o) at (7.5,2){};
     \node (r) at (.75,-.5){};
        \node (s) at (2.25,-.5){};
         \node (rr) at (.75,1.5){};
        \node (ss) at (2.25,1.5){};
         \node (rrr) at (7.5,.5){};
        \node (sss) at (9,.5){};
        \node (rrrr) at (4.5,.5){};
        \node (ssss) at (6,.5){};
        \draw[red] (r) edge (s);
        \draw[red] (rr) edge (ss);
        \draw[red] (rrr) edge (sss);
        \draw[red](rrrr) edge (ssss);
    \draw[blue] (f) edge node[below=.5mm, draw=none, fill=none, text=black] {I} (e)
        edge node[below, draw=none, fill=none, text=black] {A} (g) edge (h);
    \draw[blue] (b) edge node[below=.5mm, draw=none, fill=none, text=black] {I} (a)
        edge node[below, draw=none, fill=none, text=black] {A} (c) edge (d);
    \draw[blue] (l) edge node[above, draw=none, fill=none, text=black] {a} (g)
        edge node[left, draw=none, fill=none, text=black] {c} (i) edge
        node[above, draw=none, fill=none, text=black] {e}(j);
    \draw[blue] (j) edge node[above, draw=none, fill=none, text=black] {d}(i)
        edge node[above, draw=none, fill=none, text=black] {f}(k);
    \draw[blue] (m) edge node[below=.5mm, draw=none, fill=none, text=black] {I} (k)
        edge (o) edge node[below, draw=none, fill=none, text=black] {A}(n);
    \draw[blue] (c) edge node[above, draw=none, fill=none, text=black] {b} (i);
    \end{tikzpicture}
    \caption{OR gadget for triangular grids.}
      \label{fig:tri-or}
\end{figure}
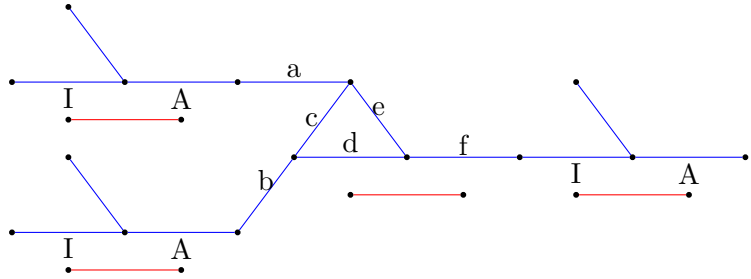
  \begin{enumerate}
         \item If the input is $(A,A)$, then edges $a$ and $b$ vanish. Blue
             should take $d$ or $e$ to clear $c,d,e$ and $f$. Then, they can
             select $A$, giving an active output;
        \item If the input is $(I,A)$, this clears $b$ but not $a$. Blue should then optimally select $e$ in order to clear $a$, which also clears $c,d$ and $f$,
            leading to an active output on the next turn;
        \item If the input is $(A,I)$, this clears $a$ but not $b$. Blue clears
            $b$ by selecting $d$, which also clears $c,e$ and $f$, leading to
            an active output on the next turn;
        \item If the input is $(I,I)$, then Blue would optimally select $c$ to clear $a$
            and $b$. This leaves $f$, so Blue must choose the inactive output.
    \end{enumerate}
Lastly, we describe a CHOICE gadget for triangular grids in Figure \ref{fig:tri-choice}.
\begin{figure}
    \centering
    \begin{tikzpicture}[every node/.style={draw,circle,inner sep=0pt,minimum
        size=1pt, very thick,fill=black}]
    \node (a) at (0,2){};
    \node (b) at (1.5,2){};
    \node (c) at (3,2){};
    \node (d) at (4.5,2){};
    \node (e) at (0.75,3){};
    \node (f) at (2.25,3){};
    \node (g) at (5.25,1){};
    \node (h) at (5.25,3){};
    \node (i) at (6,0){};
    \node (j) at (6,4){};
    \node (k) at (7.5,0){};
    \node (l) at (9,0){};
    \node (m) at (6.75,1){};
    \node (n) at (7.5,4){};
    \node (o) at (9,4){};
    \node (p) at (6.75,5){};
    \node (r) at (6.75,-.5){};
        \node (s) at (8.25,-.5){};
         \node (rr) at (6.75,3.5){};
        \node (ss) at (8.25,3.5){};
         \node (rrr) at (.75,1.5){};
        \node (sss) at (2.25,1.5){};
        \node (rrrr) at (3,1.5){};
        \node (ssss) at (4.5,1.5){};
        \draw[red] (r) edge (s);
        \draw[red] (rr) edge (ss);
        \draw[red] (rrr) edge (sss);
        \draw[red] (rrrr) edge (ssss);
    \draw[blue] (b) edge node[below=.5mm, draw=none, fill=none, text=black] {I} (a)
        edge (e) edge node[below, draw=none, fill=none, text=black] {A} (c);
    \draw[blue] (c) edge node[above, draw=none, fill=none, text=black] {a}(f)
        edge node[above, draw=none, fill=none, text=black] {b}(d);
    \draw[blue] (h) edge node[above, draw=none, fill=none, text=black] {c} (d)
        edge node[above, draw=none, fill=none, text=black] {e}(j);
    \draw[blue] (g) edge node[above, draw=none, fill=none, text=black] {d} (d)
        edge node[above, draw=none, fill=none, text=black] {f}(i);
    \draw[blue] (n) edge node[below=.5mm, draw=none, fill=none, text=black] {I} (j)
        edge node[below, draw=none, fill=none, text=black] {A} (o) edge (p);
    \draw[blue] (k) edge node[below=.5mm, draw=none, fill=none, text=black] {I} (i)
        edge node[below, draw=none, fill=none, text=black] {A} (l) edge (m);
    \end{tikzpicture}
    \caption{CHOICE gadget for triangular grids.}
        \label{fig:tri-choice}
\end{figure}
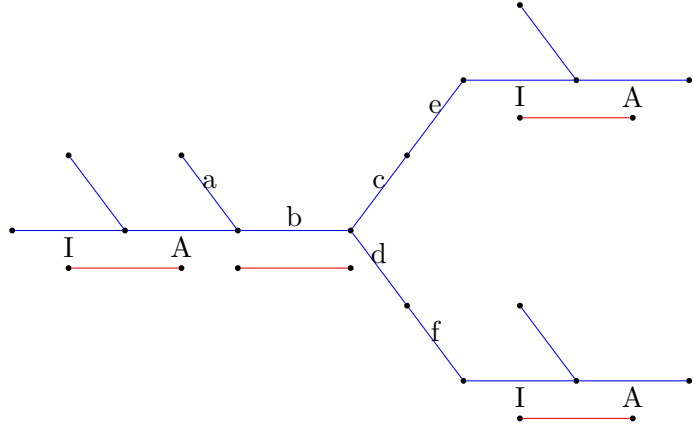
  \begin{enumerate}
         \item If the input is $A$, then edges $a$ and $b$ vanish. Blue would then choose to take $c$ or $d$. If they choose $c$, edges $c,d,e$ vanish while edge $f$ remains, leading to an active output in the upper branch and an inactive output in the lower branch; the case when they choose $d$ similarly leads to an inactive output in the upper branch and an active one in the lower branch;
        \item If the input is $I$, the edge $a$ remains, which urges $B$ to select $b$. Then edges $c$ and $d$ are cleared, leaving edges $e$ and $f$, which leads to active outputs in both branches.
    \end{enumerate}

\subsection{Line Graph Planarity}

The line graphs of all of our gadgets turn out to be planar. In particular, the
planarity of the line graph of CHOICE gadget for Cartesian grids can be
witnessed as shown in Figure \ref{fig:planar-choice-cart}.

\begin{figure}
\centering
\begin{tikzpicture}[every node/.style={draw,circle,inner sep=0pt,minimum
    size=1pt, very thick,fill=black}]
    \node (A)[label=above:A] at (0,0){};
    \node (d)[label=above:d] at (1.5,0){};
    \node (i)[label=above:i] at (3,0){};
    \node (b)[label=above:b] at (0.75,1){};
    \node (a)[label=above:a] at (0.75,-1){};
    \node (e)[label=above:e] at (1.5,2){};
    \node (c)[label=above:c] at (1.5,-2){};
    \node (f)[label=above:f] at (2.25,1){};
    \node (g)[label=above:g] at (2.25,-1){};
    \node (h)[label=above:h] at (3,2){};
    \node (j)[label=above:j] at (3,-2){};
    \draw (d) edge (b) edge (A) edge (a) edge (g) edge (f) edge (i);
    \draw (b) edge (A) edge (e);
    \draw (a) edge (A) edge (c);
    \draw (f) edge (e) edge (i);
    \draw (g) edge (c) edge (i);
    \draw (h) edge (e) edge (f);
    \draw (j) edge (c) edge (g);
    \draw[decorate, decoration={amplitude=2mm}]
  (b) .. controls (-1,0.5) and (-1,-0.5) .. (a);
    \draw[decorate, decoration={amplitude=2mm}]
  (f) .. controls (4,0.5) and (4,-0.5) .. (g);

\end{tikzpicture}
\label{fig:planar-choice-cart}
\end{figure}

\end{document}